\newif\ifconf
\conffalse %full version

\ifconf
\documentclass[envcountsect]{llncs}
\else
\documentclass[12pt,envcountsect,runningheads]{llncs}
\usepackage{fullpage}
\fi

\usepackage{amssymb,amsmath}
\usepackage{graphicx}
\usepackage{subfigure}
\usepackage{graphics}

\renewcommand{\phi}{\varphi}

\title{Planarizing an Unknown Surface}
\author{Yury Makarychev\thanks{Yury Makarychev is supported in part by the NSF Career Award CCF-1150062.}
\and Anastasios Sidiropoulos}
\institute{Toyota Technological Institute at Chicago\\
\url{yury@ttic.edu, tasos@ttic.edu}}

\begin{document}
\maketitle

\begin{abstract}
It has been recently shown that any graph of genus $g>0$ can be stochastically embedded into a distribution over planar graphs, with distortion $O(\log (g+1))$ [Sidiropoulos, FOCS 2010]. %\cite{sidiropoulos2010optimal}.
This embedding can be computed in polynomial time, provided that a drawing of the input graph into a genus-$g$ surface is given.

We show how to compute the above embedding without having such a drawing.
This implies a general reduction for solving problems on graphs of small genus, even when the drawing into a small genus surface is unknown.
To the best of our knowledge, this is the first result of this type.
\end{abstract}

%\thispagestyle{empty}
%\newpage

\section{Introduction}

The genus of a graph is a parameter that quantifies how far it is from being planar.
Informally, a graph has genus $g$, for some
$g\geq 0$, if it can be drawn without any crossings on the surface of
a sphere with $g$ additional handles (see Section~\ref{sec:prelims}).  For example, a planar graph has genus $0$, and a
graph that can be drawn on a torus has genus at most $1$.

Planar graphs exhibit properties that give rise to improved algorithmic solutions for numerous problems (see, for example \cite{Baker-planar}).
Because of their similarities to planar graphs, graphs
of small genus enjoy similar algorithmic characteristic.
More precisely, algorithms for planar graphs can usually be extended to graphs of bounded genus, with a small loss in efficiency or quality of the solution (e.g.~\cite{CEN09}).

Unfortunately, such extensions typically suffer from two main difficulties.
First, for different problems, one typically needs to develop complicated, and ad-hoc techniques.
Second, a perhaps more challenging issue is that essentially all known algorithms for graphs of small genus require that a drawing of the input graph into a small genus surface is given.
In general, computing a drawing of a graph into a surface of minimum genus is NP-hard \cite{Thomassen89,Thomassen93a}.
Moreover, the currently best-known approximation algorithm for this problem is only a trivial $O(n)$-approximation that follows by bounds on the Euler characteristic.
This has been improved to $O(\sqrt{n})$-approximation for graphs of bounded degree \cite{ChenKK97}.

The first of the above two obstacles has been recently addressed for some problems by Sidiropoulos \cite{sidiropoulos2010optimal}, who showed that any graph of genus $g>0$ can be embedded into a distribution over planar graphs, with distortion $O(\log (g+1))$ (see Section~\ref{sec:prelims} for definitions).
This result implies a general reduction for a large class of geometric optimization problems from instances on genus-$g$ graphs, to corresponding ones on planar graphs, with a $O(\log (g+1))$ loss factor in the approximation guarantee. 

Unfortunately, the algorithm from \cite{sidiropoulos2010optimal} can compute the above embedding in polynomial time, only if a drawing of the input graph into a small genus surface is given.
We show how to compute this embedding even when the drawing of the input graph is unknown.
In particular, this implies that the above reduction for solving problems on graphs of small genus, can be performed even on graphs for which we don't have a drawing into a small genus surface.
The statement of our main embedding result follows.

\begin{theorem}[Main result]\label{tim:main}
There exists a polynomial time algorithm which given a graph $G$ of genus $g>0$, computes a stochastic embedding of $G$ into planar graphs, with distortion $O(\log (g+1))$.
In particular, the algorithm does not require a drawing of $G$ as part of the input.
\end{theorem}

\iffalse
\begin{theorem}[Planarization of an unknown surface]\label{thm:main}
Any graph $G$ of genus $g$, admits a stochastic embedding into planar graphs, with distortion $O(\log g)$.
Moreover, given a drawing of $G$ into a genus-$g$ surface, the embedding can be computed in polynomial time.
\end{theorem}
\fi

\subsection{Applications}

The main application of our result is a general
reduction from a class of optimization problems on
genus-$g$ graphs, to their restriction on planar graphs.
This is the same reduction obtained in \cite{sidiropoulos2010optimal},
only here we don't require a drawing of the input graph.
For completeness, we state precisely the reduction, as given in \cite{sidiropoulos2010optimal} (see also \cite{Bar96}).
Let $V$ be a set,
${\cal I}\subset \mathbb{R}_+^{V\times V}$ a set of non-negative vectors corresponding to all feasible solutions for a minimization problem, and $c \in \mathbb{R}_+^{V\times V}$.
Then, we define the \emph{linear minimization problem} $({\cal I}, c)$ to be the computational problem where we are given a graph $G=(V,E)$, and we are asked to find $s\in {\cal I}$, minimizing
\[
\sum_{\{u,v\}\in V\times V} c_{u,v} \cdot s_{u,v} \cdot d(u,v)
\]

Observe that this definition captures a very general class of problems.
For example, MST can be encoded by letting ${\cal I}$ be the set of indicator vectors of the edges of all spanning trees on $V$, and $c$ the all-ones vector.
Similarly, one can easily encode problems such as TSP, Facility-Location, $k$-Server, Bi-Chromatic Matching, etc.

The main Corollary of our embedding result can now be stated as follows.

\begin{corollary}\label{cor:opt}
Let $\Pi=({\cal I}, c)$ be a linear minimization problem.  If there exists a polynomial-time $\alpha$-approximation algorithm for $\Pi$ on planar graphs, then there exists a randomized polynomial-time $O(\alpha\cdot \log (g+1))$-approximation algorithm for $\Pi$ on graphs of genus $g>0$, even when the drawing of the input graph is unknown.
\end{corollary}

\subsection{Overview of the Algorithm}
We now give a high-level overview of our algorithm.
Consider a graph $G=(V,E)$.
Let us say that a collection ${\cal P}$ of shortest paths in $G$ is a \emph{planarizing set of paths}, if the graph $G\setminus \bigcup_{P\in {\cal P}} V(P)$ is planar.
It was shown by Sidiropoulos \cite{sidiropoulos2010optimal} that any graph having a planarizing set of paths of size $k$, admits a stochastic embedding into planar graphs, with distortion $O(\log k)$.
Moreover, given such a set of planarizing paths, the embedding can be computed in polynomial time.
It follows by the work of Eppstein \cite{eppstein2003dynamic}, and Erickson and Whittlesey \cite{erickson2005greedy}, that for any graph $G$ of genus $g$, that there exists a planarizing set of paths, of size $O(g)$.
However, all known algorithms for computing this planarizing set require a drawing of the graph into a surface of genus $g$.
Since we don't know how to compute a drawing of a graph into a minimum-genus surface in polynomial time, all known algorithms are not applicable in our case.

Our main technical contribution is showing how to compute in polynomial time a planarizing set of paths 
of approximately optimal size (up to a $\operatorname{polylog} n$ factor) in an arbitrary graph.
For a graph $G$, we say that a collection ${\cal Q}$ of shortest paths having a common endpoint is a \emph{balanced set of paths} if $\bigcup_{Q\in {\cal Q}} V(Q)$ is a balanced vertex-separator of $G$. That is, removing all paths in ${\cal Q}$ from $G$, leaves a graph where every connected component is at most half the size of $G$.
Our high-level approach is as follows. We find and remove a ``small'' balanced set of paths in $G$. 
Then we compute connected components in the obtained graph.
In each non-planar connected component, we again find and remove a balanced set of paths. We 
repeat this procedure until all components are planar. Finally, we output the planarizing set of paths 
that consists of all paths that we removed from the graph.

In order for this approach to work, we first prove that in a (possibly vertex-weighted) graph $G$ of genus $g$, 
there exists a balanced set $Q$ of paths of size $O(g)$.
Next, we show how to compute in polynomial time a balanced set of paths of approximately optimal size in an arbitrary graph $G$.
As outlined above, we then recursively use this as a subroutine to find a set ${\cal P}$ of planarizing paths.
We begin with a graph $G$ of genus $g$ (for which we don't have a drawing into a genus-$g$ surface), and 
inductively build ${\cal P}$ in steps.
At the first step, we compute a balanced set ${\cal Q}_1$ of paths in $G$.
We add these paths to ${\cal P}$.
At every subsequent step $i>1$, let $G_i$ be the graph obtained from $G$ after removing all the paths we have computed so far, i.e.~$G_i=G\setminus \bigcup_{P\in {\cal P}} V(P)$.
Since $G$ has genus $g$, graph $G_i$ has at most $O(g)$ non-planar connected components.
For every such non-planar component, we compute a balanced set of paths and add it to ${\cal P}$.
We show that after every step, the size of the largest non-planar component reduces by at least a constant factor.
Therefore, after $O(\log n)$ steps, we obtain the desired planarizing set of paths.

\subsection{Related Work}

Inspired by Bartal's stochastic embedding of general metrics into
trees \cite{Bar96}, Indyk and Sidiropoulos~\cite{indyk_genus} showed that every metric on a graph of genus $g$ can be stochastically embedded into a planar graph with distortion $2^{O(g)}$
(see Section \ref{sec:prelims} for a formal definition of
stochastic embeddings).
The above bound was later improved by Borradaile, Lee, and Sidiropoulos \cite{BLS09}, who obtained an embedding with distortion $g^{O(1)}$.
Subsequently, Sidiropoulos \cite{sidiropoulos2010optimal} gave an embedding with distortion $O(\log g)$, matching the $\Omega(\log g)$ lower bound from \cite{BLS09}.
The embeddings from \cite{indyk_genus}, and \cite{sidiropoulos2010optimal} can be computed in polynomial time, provided that the drawing of the graph into a small genus surface is given.
Computing the embedding from \cite{BLS09} requires solving an NP-hard problem, even when the drawing is given.

\subsection{Preliminaries}\label{sec:prelims}
Throughout the paper, we consider graphs
with non-negative edge lengths.
For a tree $T$ with root $r\in V(T)$, and for $v\in V(T)$ we denote by $T(v)$ the unique path in $T$ between $v$ and $r$.

\paragraph{Graphs on surfaces}
Let us recall some notions from topological graph theory (an in-depth
exposition can be found in \cite{MoharT-book}).  A \emph{surface} is a
compact connected 2-dimensional manifold, without boundary.
For a graph $G$ we can
define a one-dimensional simplicial complex $C$ associated with $G$ as
follows: The $0$-cells of $C$ are the vertices of $G$, and for each
edge $\{u,v\}$ of $G$, there is a $1$-cell in $C$ connecting $u$ and
$v$.  A \emph{drawing} of $G$ on a surface $S$ is a continuous injection
$f:C\rightarrow S$.
The \emph{genus} of a surface ${\cal S}$ is the maximum cardinality of a collection
of simple closed non-intersecting curves $C_1,\dots,C_k$ in $\cal S$, such that ${\cal S} \setminus (C_1 \cup \dots \cup C_k)$ is connected.
The genus of a graph $G$ is the minimum $k$, such that $G$ can be drawn into a surface of genus $k$.
Note that a graph of genus $0$ is a planar graph.
We remark that we make no distinction between orientable, and non-orientable genus, since all of our results hold in both settings.

\paragraph{Metric embeddings}
A mapping $f : X \to Y$ between two metric spaces $(X,d)$ and $(Y,d')$
is {\em non-contracting} if $d'(f(x),f(y)) \geq d(x,y)$ for all $x,y \in X$.
If $(X,d)$ is any finite metric space, and $\mathcal Y$
is a family of finite metric spaces, we say that {\em $(X,d)$ admits a stochastic $D$-embedding into $\mathcal Y$} if there exists a random metric space $(Y,d') \in \mathcal Y$ and a random
non-contracting mapping $f : X \to Y$ such that for every $x,y \in X$,
\begin{equation}
\label{eq:expansion}
\mathbb E\left[\vphantom{\bigoplus} d'(f(x),f(y))\right] \leq D \cdot d(x,y).
\end{equation}
The infimal $D$ such that \eqref{eq:expansion} holds is the {\em distortion of
the stochastic embedding.}
A detailed exposition of  results on metric embeddings can be found in \cite{I-survey} and \cite{Matousek-book}.

\section{Path Separators in Embedded Graphs}

For a graph $G$, a real $\alpha\in (0,1/2]$, and a set $X\subseteq V(G)$ we say that $X$ is an \emph{$\alpha$-balanced vertex separator} for $G$ if every connected component of $G\setminus X$ contains at most $\alpha\cdot |V(G)|$ vertices.
It is also called simply balanced vertex separator, when $\alpha=1/2$.

For a vertex-weighted graph $G$ with weight function $w:V(G)\to \mathbb{R}_{\geq 0}$, for every $Y\subseteq V(G)$ we use the notation $w(Y) = \sum_{v\in V(G)} w(v)$.
Similarly to the unweighted case, we say that a set $X\subseteq V(G)$ is a balanced vertex separator for a weighted graph
 $(G,w)$ if for every connected component $C$ of $G\setminus X$ we have $w(V(C)) \leq w(V(G)) / 2$. 

\begin{theorem}[Lipton \& Tarjan \cite{lipton1979separator}, Thorup \cite{thorup2004compact}]\label{thm:planar_separators}
Let $G$ be a planar graph, let $r\in V(G)$, and let $T$ be a spanning tree of $G$ with root $r$.
Then, there exist $v_1, v_2, v_3\in V(G)$, such that $V(T(v_1)\cup T(v_2) \cup T(v_3))$ is a balanced vertex separator for $G$.
Moreover, the vertices $v_1, v_2$ and $v_3$ can be computed in polynomial time.
\end{theorem}

We will use a slight modification of Theorem \ref{thm:planar_separators}, for the case of weighted graphs.
The proof is a straightforward extension to the one due to Thorup \cite{thorup2004compact}, which is based on the argument of Lipton and Tarjan \cite{lipton1979separator}.

\begin{lemma}\label{lem:weighted_planar_separators}
Let $G$ be a planar graph, let $r\in V(G)$, and let $T$ be a spanning tree of $G$ with root $r$.
Let $w:V(G) \to \mathbb{R}_{\geq 0}$.
Then, there exist $v_1, v_2, v_3\in V(G)$, such that $V(T(v_1)\cup T(v_2) \cup T(v_3))$ is a balanced vertex separator for $(G,w)$.
Moreover, the vertices $v_1, v_2$ and $v_3$ can be computed in polynomial time.
\end{lemma}

The next Theorem follows by the work of Eppstein \cite{eppstein2003dynamic}, and Erickson \& Whittlesey \cite{erickson2005greedy}.

\begin{theorem}[Erickson \& Whittlesey \cite{erickson2005greedy}, Eppstein \cite{eppstein2003dynamic}]\label{thm:planarization_spanning}
Let $G$ be a graph of genus $g>0$, and let $\phi$ be an embedding of $G$ into a surface ${\cal S}$ of genus $g$.
Let $r\in V(G)$, and let $T$ be a spanning tree of $G$ with root $r$.
Then, there exist edges $\{x_1,y_1\},\ldots,\{x_{2g},y_{2g}\}\in E(G)$, such that
$G \setminus \bigcup_{i=1}^{2g} V(T(x_i)\cup T(y_i))$ is planar.
Moreover, the topological space ${\cal S} \setminus \bigcup_{i=1}^{2g} \phi(T(x_i)\cup T(y_i)\cup \{x_i,y_i\})$ is homeomorphic to an open disk.
\end{theorem}

We are now ready to prove the main result of this section.

\begin{lemma}[Existence of path separators in embedded graphs]\label{lem:genus_separator_weighted}
Let $G$ be a weighted graph of genus $g$, with weight function $w:V(G)\to \mathbb{R}_{\geq 0}$.
Let $r\in V(G)$, and let $T$ be a spanning tree of $G$ with root $r$.
Then, there exists $X\subseteq V(G)$, with $|X|\leq 4g+3$, such that $\bigcup_{u\in X} V(T(u))$ is a balanced vertex separator for $(G,w)$.
\end{lemma}
\begin{proof}
The case $g=0$ follows by Lemma \ref{lem:weighted_planar_separators}, so we may assume that $g>0$.
Fix an embedding $\phi$ of $G$ into a surface ${\cal S}$ of genus $g$.
By Theorem \ref{thm:planarization_spanning} there exist $\{x_1,y_1\},\ldots,\{x_{2g},y_{2g}\}\in E(G)$, such that the topological space ${\cal S}\setminus \bigcup_{i=1}^{2g} \phi(T(x_i)\cup T(y_i)\cup \{x_i,y_i\})$ is homeomorphic to an open disk.
Let
\[
H = \bigcup_{i=1}^{2g} T(x_i)\cup T(y_i)\cup \{x_i,y_i\}.
\]
Note that $r\in V(H)$.
Let $G'$ be the graph obtained from $G$ by contracting $H$ into a single vertex $r'$.
Since ${\cal S}\setminus \phi(H)$ is an open disk, it follows that $G'$ is planar.

Let $T'$ be the subgraph of $G'$ induced by $T$ after contracting $H$.
Since $T$ is a spanning subgraph of $G$, it follows that $T'$ is a spanning subgraph of $G'$.
%Moreover
Indeed, the set of vertices $V(H)$ spans is a connected subtree of $T$.
Therefore, after contracting $H$, the subgraph $T'$ induced by $T$ is still a tree.
Thus, $T'$ is a spanning subtree of $G'$.
We consider $T'$ being rooted at $r'$.

Define a weight function $w':V(G')\to \mathbb{R}_{\geq 0}$ such that for every $v\in V(G')$,
\[
w'(v) = \left\{ \begin{array}{ll}
  w(v), & \mbox{ if } v\neq r'\\
  0, & \mbox{ if } v=r'
\end{array}
\right.
\]
By Lemma \ref{lem:weighted_planar_separators} it follows that there exist $v_1,v_2,v_3\in V(G')$ such that $V(T'(v_1)\cup T'(v_2)\cup T'(v_3))$ is a balanced vertex separator for $(G',w')$.

Let $J = G\setminus V(H)$.
Observe that $J = G\setminus V(H) = G' \setminus \{r'\}$.
Moreover, for any $v\in V(J)$, we have $T(v)\cap J = T'(v)\cap J$.
Thus, the set of connected components of $(G \setminus V(H)) \setminus V(T(v_1)\cup T(v_2)\cup T(v_3))$ is the same as the set of connected components of $(G' \setminus \{r'\}) \setminus V(T'(v_1)\cup T'(v_2)\cup T'(v_3))$.
Let $C$ be a connected component of $(G \setminus V(H)) \setminus V(T(v_1)\cup T(v_2)\cup T(v_3))$.
We have
\[
w(C) = w'(C) \leq \frac{1}{2} w(V(G')) = \frac{1}{2} (w(V(G)) - w(V(H))) \leq \frac{1}{2} w(V(G)).
\]
Thus, $V(T(v_1) \cup T(v_2)\cup T(v_3)) \cup \bigcup_{i=1}^{2g} V(T(x_i) \cup T(y_i))$ is a balanced vertex separator for $(G,w)$, as required.
\qed
\end{proof}

\section{Computing Path Separators in Arbitrary Graphs}
Recall the definition of a 
caterpillar decomposition of a tree.
\begin{definition}[Caterpillar decomposition \cite{matousek1999trees,charikar2002dimension}]
A \emph{caterpillar decomposition}
of a rooted tree $T$ is a family of paths ${\cal P} = \{P_i\}$, satisfying the following conditions:
\begin{description}
\item{(i)}
Every $P_i\in {\cal P}$ is a subpath of a root--leaf path.
\item{(ii)}
For every $P_i\neq P_j\in {\cal P}$, we have $V(P_i)\cap V(P_j)=\emptyset$.
\item{(iii)}
$V(T) = \bigcup_{P_i\in {\cal P}} V(P_i)$.
\end{description}
\end{definition}
The proof of the following lemma about caterpillar decompositions can be found in \cite{matousek1999trees,charikar2002dimension}.

\begin{lemma}[See \cite{matousek1999trees,charikar2002dimension}]\label{lem:computing_caterpillar}
For every rooted tree $T$, there exists a caterpillar decomposition ${\cal P}$, such that every root--leaf path $T(u)$ crosses
at most $O(\log n)$ paths from $\cal P$.
Moreover, this decomposition can be found in polynomial time.
\end{lemma}

We are now ready to prove that main result of this section.

\begin{lemma}[Computing approximate path separators]\label{lem:genus_separator_weighted_approximate}
Let $G$ be a graph, and $w:V(G)\to \mathbb{R}_{\geq 0}$.
Let $r\in V(G)$, and let $T$ be a spanning tree of $G$ with root $r$.
Suppose that there exists $X\subseteq V(G)$, such that $\bigcup_{u\in X} V(T(u))$ is a balanced vertex separator for $(G,w)$.
Then we can compute in polynomial time a set $Y\subseteq V(G)$ 
with $|Y| \leq O(\log^{3/2} n) \cdot |X|$, such that $\bigcup_{u\in Y} V(T(u))$ is a $3/4$-balanced vertex separator for $(G,w)$.
\end{lemma}
\begin{proof}
We reduce the problem to the problem of finding a vertex separator in an auxiliary graph.
Using Lemma \ref{lem:computing_caterpillar} we construct a caterpillar decomposition ${\cal P}$ of $T$ such that every root--leaf path $T(u)$ crosses
at most $O(\log n)$ paths from $\cal P$. We define an auxiliary graph ${\cal G}$ on the set $\cal P$ as follows:
$P_i\in \cal P$ and $P_j \in \cal P$ are connected with an edge in ${\cal G}$ if there is an edge between sets
$V(P_i)$ and $V(P_j)$ in $G$. We assign each $P_i$ weight equal to the total weight of all 
vertices of $P_i$. Note that then the total weight of all vertices in $\cal G$ equals $w(V(G))$.

Observe that for every ${\cal A}\subset {\cal P}$ the induced graph ${\cal G}[{\cal A}]$ is connected if and only if 
the induced graph $G[A]$, where  $A=\bigcup_{P_i\in {\cal A}} V(P_i)$, is connected. Consequently,
if ${\cal C}_1, \dots,{\cal C}_t$ are connected components of ${\cal G} \setminus {\cal B}$
(for some ${\cal B}\subset {\cal P}$) then sets $C_j = \bigcup_{P_i\in {\cal C}_j} V(P_i)$ (for $j=1,\dots,t$) are connected 
components of $G \setminus B$ where $B=\bigcup_{P_i\in {\cal B}} V(P_i)$; moreover, the weight of each ${\cal C}_i$ equals
the weight of $C_i$.
Therefore, $\cal B$ is a balanced vertex separator in $\cal G$ if and only if 
$B=\bigcup_{P_i\in {\cal B}} V(P_i)$ is a balanced vertex separator in $G$.

We now prove that there is a balanced vertex separator in $\cal G$ of size
$O(\log n) \cdot |V(G)|$. Let ${\cal X} = \bigcup_{u\in X} \{P_i\in {\cal P}: P_i \text{ intersects } T(u)\}$.
First, we show that ${\cal X}$ is a balanced vertex separator in $\cal G$.
Denote $X' = \bigcup_{P_i\in {\cal X}} V(P_i)$. Observe that $X' \supset \bigcup_{u\in X} V(T(u))$. Indeed,
consider $v\in \bigcup_{u\in X} V(T(u))$. Then $v\in T(u)$ for some $u\in X$. Let $P_i$ be the path in $\cal P$ that contains $v$. 
Then $P_i$ intersects $T(u)$ at vertex $v$ and therefore $P_ i\in {\cal X}$. Hence $v \in V(P_i) \subset X'$.
We conclude that $X' \supset \bigcup_{u\in X} V(T(u))$. Since $\bigcup_{u\in X} V(T(u))$ is a balanced vertex separator in $G$,
set $X'$ is also a balanced vertex separator in $G$. Hence $\cal X$ is a balanced vertex separator in $\cal G$.
Now we upper bound the size of $\cal X$. Note that for every $u$, we have
$|\{P_i\in {\cal P}: P_i \text{ intersects } T(u)\}| = O(\log n)$ (by Lemma~\ref{lem:computing_caterpillar}). Thus we have,
$|{\cal X}| = O(\log n) \cdot |X|$. We proved that there is a balanced vertex separator in $\cal G$ of size
$O(\log n) \cdot |X|$. 

We use the algorithm of Feige, Hajiaghayi and Lee~\cite{FHL08} to find 
a $O(\sqrt{\log n})$ approximation for the optimal balanced vertex separator in $\cal G$. 
We get a $3/4$-balanced vertex separator ${\cal Y} \subset {\cal P}$ in $\cal G$ of size  at most
$O(\sqrt{\log n}) \cdot |{\cal X}| = O(\log^{3/2} n) \cdot |X|$.

Finally, we define the set $Y$. 
For every path $P_i \in {\cal P}$, let $p_i$ be a leaf of $T$ such that $P_i$ is a subset of $T(p_i)$. 
Let $Y= \{p_i: P_i \in {\cal P}\}$. 
Note that  $|Y| \leq |{\cal Y}| = O(\log^{3/2} n) \cdot |X|$.
Since $\cal Y$ is a $3/4$-balanced separator in $\cal G$, the set
$Y' = \bigcup_{P_i \in {\cal Y}} V(P_i)$ is a $3/4$-balanced 
separator in $G$, and therefore $\bigcup_{u\in Y} V(T(u)) \supset Y'$ 
is a $3/4$-balanced separator in $G$. 
\qed
\end{proof}

\section{Computing Planarizing Sets of Paths}

\begin{lemma}[Computing a planarizing set of paths]\label{lem:computing_planarizing_paths}
Let $G$ be an $n$-vertex graph of genus $g>0$.
Let $r\in V(G)$, and let $T$ be a spanning subtree of $G$ with root $r$.
Then, we can compute in polynomial time a set $X\subseteq V(G)$, with $|X| = O(g^{2}\cdot \log^{5/2} n)$, such that the graph $G \setminus \bigcup_{v\in X} V(T(v))$ is planar.
\end{lemma}
\begin{proof}
We inductively construct a sequence $\{X_i\}_{i=0}^k$, for some $k=O(\log n)$, where for every $i\in \{0,\ldots,k\}$, we have $X_i\subseteq V(G)$.
The resulting desired set will be $X=\bigcup_{i=0}^k X_i$.

For the basis of the induction, we set $X_0=\emptyset$.

Let $i>0$, and suppose that $X_{i-1}$ has already been constructed.
We show how to construct $X_i$.
Let ${\cal C}_i$ be the set of connected components of $G\setminus \bigcup_{j=0}^{i-1} \bigcup_{u\in X_j} V(T(u))$.
Let also ${\cal C}_i'$ be the set of non-planar components in ${\cal C}_i$.
Note that $G$ is the only component in ${\cal C}_1$.
For every component $C\in {\cal C}_i$ we define a function $w_C:V(G)\to \mathbb{R}_{\geq 0}$ such that for every $v\in V(G)$,
\[
w_C(v) = \left\{
\begin{array}{ll}
  1, & \mbox{ if } v\in V(C)\\
  0, & \mbox{ if } v\notin V(C)
\end{array}
\right.
\]
By Lemma \ref{lem:genus_separator_weighted} it follows that there exists $Y_C\subseteq V(G)$, with $|Y_C| \leq 4g+3$, such that $\bigcup_{v\in Y_C}V(T(v))$ is a balanced vertex separator for $(G, w_C)$.
Therefore, by Lemma \ref{lem:genus_separator_weighted_approximate} we can compute in polynomial time a set $Z_C\subseteq V(G)$, with 
\[
|Z_C| \leq O(\log^{3/2} n)\cdot |Y_C| \leq O(\log^{3/2} n) \cdot (4g+3),
\]
and such that $\bigcup_{v\in Z_C}V(T(v))$ is an $3/4$-balanced vertex separator for $(G, w_C)$.
We set
\[
X_i = \bigcup_{C\in {\cal C}_i'} Z_C.
\]
This concludes the inductive construction of the sequence $\{X_i\}_{i=0}^k$.

We next show that for some $k=O(\log n)$, the set $X=\bigcup_{i=0}^k X_i$ is as required.
Consider some $i\geq 1$, and let $C\in {\cal C}_i'$ be a non-planar connected component of $G\setminus \bigcup_{j=0}^{i-1} \bigcup_{u\in X_j} V(T(u))$.
Observe that there exists a connected component $C'\in {\cal C}_{i-1}$ such that $C\subseteq C'$.
Since $C$ is non-planar, it follows that $C'$ is also non-planar, and thus $C'\in {\cal C}_{i-1}'$.
By the construction, the set $X_i$ contains the set $Z_{C'}$, where $\bigcup_{v\in Z_{C'}}V(T(v))$ is 
a $3/4$-balanced vertex separator for $(G, w_{C'})$.
It follows that $|V(C)| \leq 3|V(C')|/4$.
Thus, the size of every non-planar connected component in ${\cal C}_i$ is at most $(3/4)^{i-1}|V(G)|$.
This implies in particular that for $k=\lceil \log_{4/3} n\rceil-1$, the set ${\cal C}_k$ does not contain any non-planar connected components, and therefore the graph $G \setminus \bigcup_{v\in X} V(T(v))$ is planar.

It remains to upper bound $|X|$.
Since $G$ has genus $g$, we have that for every $i\in \{0,\ldots,k\}$, the set ${\cal C}_i$ contains at most $g$ non-planar connected components, i.e.~$|{\cal C}_i'|\leq g$.
Therefore,
\[
|X| \leq \sum_{i=0}^k |X_i| \leq \sum_{i=0}^k \sum_{C\in {\cal C}_i'} |Z_C| \leq \sum_{i=0}^k \sum_{C\in {\cal C}_i'} O(\log^{3/2} n)\cdot (4g+3) \leq O(g^2 \cdot \log^{5/2} n),
\]
as required.
\qed
\end{proof}

\section{Putting Everything Together}

The next lemma follows by the work of Sidiropoulos \cite{sidiropoulos2010optimal}.

\begin{lemma}[Sidiropoulos \cite{sidiropoulos2010optimal}]\label{lem:optimal}
Let $G$ be a graph, and $r\in V(G)$.
Let $P_1,\ldots,P_k$ be a collection of shortest paths in $G$, having $r$ as a common end-point.
Suppose that $G \setminus \bigcup_{i=1}^k V(P_i)$ is planar.
Then, $G$ admits a stochastic embedding into planar graphs, with distortion $O(\log k)$.
Moreover, if the paths $P_1,\ldots, P_k$ are given, then we can sample from the stochastic embedding in polynomial time.
\end{lemma}

\begin{theorem}[Kawarabayashi, Mohar \& Reed \cite{DBLP:conf/focs/KawarabayashiMR08}]\label{thm:genus_linear}
There exists an algorithm which given a graph $G$ of genus $g$, computes a drawing of $G$ into a surface of genus $g$, in time $O\left(2^{O(g)}\cdot n\right)$.
\end{theorem}

\begin{theorem}[Main result]
There exists a polynomial time algorithm which given a graph $G$ of genus $g>0$, computes a stochastic embedding of $G$ into planar graphs, with distortion $O(\log (g+1))$.
In particular, the algorithm does not require a drawing of $G$ as part of the input.
\end{theorem}
\begin{proof}
We can use the algorithm of Kawarabayashi, Mohar \& Reed from Theorem \ref{thm:genus_linear} to test whether $g\leq \log n$ in polynomial time.
If $g\leq \log n$, then the algorithm from Theorem \ref{thm:genus_linear} returns a drawing of $G$ into a surface of genus $g$.
Since we have a drawing of $G$ into a surface of genus $g$, we can use the algorithm of Sidiropoulos \cite{sidiropoulos2010optimal}, to compute the required embedding.

Otherwise, if $g>\log n$, we proceed as follows.
Let $r$ be an arbitrary vertex in $G$, and let $T$ be a shortest-path tree in $G$, with root $r$.
By Lemma \ref{lem:computing_planarizing_paths} we can compute a set $X\subseteq V(G)$, with $|X|=O(g^2\cdot \log^{5/2} n) = O(g^{9/2})$, such that the graph $G\setminus \bigcup_{v\in X} V(T(v))$ is planar.
Since for every $v\in X$, the path $T(v)$ has $r$ as an endpoint, it follows that we can use Lemma \ref{lem:optimal} with the collection of paths $\{T(v)\}_{v\in X}$, to compute in polynomial time a stochastic embedding into planar graphs, with distortion $O(\log |X|) = O(\log (g+1))$, as required.
\qed
\end{proof}

\bibliography{bibfile}
\bibliographystyle{alpha}

\end{document}